\documentclass[aps,prd,preprint,groupedaddress,showpacs,floatfix]{revtex4}

\usepackage{amsfonts}
\usepackage{graphicx}
\usepackage{slashed}

\newtheorem{theorem}{Theorem}

\def\endprf{\hfill  {\vrule height6pt width6pt depth0pt}\medskip}
\newenvironment{proof}{\noindent {\bf Proof} }{\endprf\par}

\begin{document}


\title{Mass generation and supersymmetry}


\author{Marco Frasca}
\email[]{marcofrasca@mclink.it}
\affiliation{Via Erasmo Gattamelata, 3 \\ 00176 Roma (Italy)}


\date{\today}

\begin{abstract}
Using a recent understanding of mass generation for Yang-Mills theory and a quartic massless scalar field theory mapping each other, we show that when such a scalar field theory is coupled to a gauge field and Dirac spinors, all fields become massive at a classical level with all the properties of supersymmetry fulfilled, when the self-interaction of the scalar field is taken infinitely large. Assuming that the mechanism for mass generation must be the same in QCD as in the Standard Model, this implies that Higgs particle must be supersymmetric.
\end{abstract}


\maketitle


\section{Introduction}

Standard model requires that all the particles entering into the theory must be massless. So, a mechanism must exist that generates these masses and breaks in some way the symmetry of the model. A model that grants this is the Englert-Brout-Higgs-Guralnik-Hagen-Kibble mechanism. The essential idea behind this mechanism is to make a field interact with a scalar field having minims shifted from zero \cite{eng,hig1,hig2,gur}. This idea has been so successful to being an acquired part of theoretical physics with a wealth of applications. The best of these has been its insertion into the standard model making this a consistent theory to explain all known particle phenomenology\cite{sm}.

A lot of activities has currently done to understand Yang-Mills theory in the low energy limit. Most of them are performed solving the theory on the lattice to understand both the spectrum and the behavior of the two-point functions\cite{cuc,ste,tep,mor} but also using functional techniques trying to unveil the solutions of Dyson-Schwinger equations \cite{avs,an,bou2,bou3}.

In this context we have solved classical Yang-Mills equations of motion showing that, in a limit of the gauge coupling going to infinity, the theory admits massive solutions already classically. The theory can then be extended to a quantum analysis maintaining such a property\cite{fra2,fra3}. These developments become possible after we showed that, for a quartic scalar field theory, a strongly coupled quantum field theory can be devised \cite{fra1}.

Our aim in this paper is to prove a theorem for massless quantum electrodynamics with a quartic scalar field. We prove that, in a strong coupling limit for the scalar field, all the fields of the model acquire the same mass and the coupling are fixed consistently with a supersymmetric model. We maintain for the present the analysis to classical solutions but this model is amenable to a quantum treatment. A recent analysis also shows that the coupling of the scalar field theory decreases at lower momenta \cite{fra4,sus1,sus2,pod} making all the scenario consistent to be extended to the standard model.

The paper is structured as follows. In sec.\ref{sec1} we describe the mechanism of generation of the mass that applies also to Yang-Mills theory. In sec.\ref{sec2} we prove the theorem showing how massless scalar field theory with a quartic scalar field, in the limit of an infinitely large self-interaction of the scalar field, gets all the fields massive and the couplings consistent with a supersymmetric theory. In sec.\ref{sec3} we make some considerations about this mechanism and the standard model. Finally, in sec.\ref{sec4} we present some conclusions.

\section{Mechanism of mass generation\label{sec1}}

In ref.\cite{fra1} we showed how a mass gap can arise for a scalar field with a quartic self-interaction. This is due to the coupling going to infinity that makes self-interaction strong. We would like to make this paper self-contained, so we consider the case of a complex scalar field with Lagrangian
\begin{equation}
   L=\partial_\mu\phi^\dagger\partial^\mu\phi-\frac{\lambda}{2}|\phi|^4. 
\end{equation}
This field is massless but, already at a classical level, displays a massive excitation. This can be easily seen from the motion equation given by
\begin{equation}
   \partial^2\phi+\lambda|\phi|^2\phi=0.
\end{equation}
This equation admits the simple solution
\begin{equation}
\label{eq:exsol}
   \phi(x)=\mu\left(\frac{2}{\lambda}\right)^\frac{1}{4}e^{i\alpha}{\rm sn}(p\cdot x+\theta,i)
\end{equation}
being $\alpha$ and $\theta$ two arbitrary phases, $\mu$ an integration constant and sn a Jacobi elliptic function. This holds provided the following dispersion relation holds
\begin{equation}
\label{eq:disp0}
   p^2=\mu^2\sqrt{\frac{\lambda}{2}}
\end{equation}
that is, this solution describes a free massive solution. That this is a free particle can be also seen from the fact that this solution admits an expansion in plane waves. So, we see here that increasing the strength of the non-linearity in the equation just conspires in producing a mass.

It is important to note that these solution have not infinite energy. As one can see by direct substitution that these represent a set of exact solutions with finite energy and this energy is the one of a free massive particle. Nonlinearity conspires to produce a mass to the excitations of the field. On the other side, if one looks at the Hamiltonian of the theory given by
\begin{equation}
   H = \int d^3x \left[\frac{1}{2}|\partial_t\phi|^2+\frac{1}{2}|\nabla\phi|^2+\frac{\lambda}{4}|\phi|^4\right]
\end{equation}
a direct substitution of the solution (\ref{eq:exsol}) gives an infinity exactly as one could obtain in a theory of free non-interacting particles described by ordinary plane-waves. To evade this problem, one generally considers a box with a side $L$ and takes the limit $L\rightarrow\infty$ at the end of computation. In our case, we need to observe that the Jacobi function ${\rm sn}$ has a real period $4nK(i)$ being $n$ an integer and $K(i)=\int_0^{\pi/2}d\theta/\sqrt{1+\sin^2\theta}=1.3111\ldots$ an elliptic integral. This means that the solution in the box is the same as (\ref{eq:exsol}) provided we quantize the momentum as
\begin{equation}
   p_k=\frac{4n_kK(i)}{L}
\end{equation}
with $k=x,y,z$. This gives without difficulty
\begin{equation}
\label{eq:H}
   H=\frac{\mu^2}{2}\sqrt{\frac{2}{\lambda}}V\left(p_0^2+\frac{1}{3}{\bf p}^2\right).
\end{equation}
With respect to eq.(\ref{eq:exsol}) that has a finite energy given by eq.(\ref{eq:disp0}) we have here a couple of inconsistencies. The first one arises by taking the limit $\lambda\rightarrow 0$ that does not recover the free massless limit and the other one arises when the volume $V$ is taken to be infinitely large as we need. Both the inconsistencies are removed by observing that $\lambda$ is a free parameter of the theory and can be arbitrarily rescaled through other arbitrary parameters. So, if we normalize eq.(\ref{eq:H}) to eq.(\ref{eq:disp0}), one can rescale
\begin{equation}
 \lambda = \bar\lambda(p)V^2\mu^6,
\end{equation}
with $\mu$ introduced by dimensional reasons, showing a running coupling. This result says that, being $\lambda$ an arbitrary parameter of the theory one can always rescale it to prove that our solutions have finite energy as eq.(\ref{eq:disp0}) is rightly stating. This procedure is quite similar to the case of a wave equation to force finite energy in an infinite volume. Taking this limit directly on the solutions has no meaning in this technique but works for the Hamiltonian.

This mechanism appears at work in Yang-Mills theory as we proved recently\cite{fra2,fra3}. Indeed, on a basis of a mapping theorem whose proof has been completed in \cite{fra3}, it is possible to show that a solution of the classical SU(N) Yang-Mills equations of motion, in the limit of the gauge coupling going to infinity, can be written as
\begin{equation}
   A_\mu^a=\eta_\mu^a\phi(x)+O\left(\frac{1}{\sqrt{N}g}\right)
\end{equation}
being $\eta_\mu^a$ a chosen set of constants (Smilga's choice\cite{smi}) and $g$ the gauge coupling. So, these gauge connections describe a free massive particle provided we do the substitution $\lambda\rightarrow Ng^2$ where we can recognize the 't~Hooft coupling of strong interactions.

The question we are interested here is if such a mechanism to generate mass can be identically used for other field theories and specially for the standard model. One should expect that, in the limit of the coupling going to infinity, also the fields interacting with such a scalar field will acquire a mass. Such a scalar field can be considered a Higgs field yet but we do not assume it has a mass term. We will see that this is realized if the essential requirements of supersymmetry are satisfied.

\section{Massless quantum electrodynamics\label{sec2}}

The simplest case where a model for mass generation can be tested is the U(1) case of quantum electrodynamics. So, we write down the following Lagrangian
\begin{equation}
   L=-\frac{1}{4}F^2+\bar\psi\left(i\slashed{\partial}-ie\slashed{A}-g\phi\right)\psi
   +\left|\partial\phi-ieA\phi\right|^2-\frac{\lambda}{2}|\phi|^4
\end{equation}
and we take the limit $\lambda\rightarrow\infty$ . Here, $e$ is the gauge coupling, $g$ is the strength of the Yukawa coupling and $\phi=\phi_1+i\phi_2$ a complex scalar field. We note that in this model all the particles are massless. We expect that they get their masses through the interaction with the scalar field that by itself becomes massive by a strong self-interaction. We prove the following theorem
\begin{theorem}[Mass generation]
\label{teo1}
In the limit $\lambda\rightarrow\infty$, quantum electrodynamics interacting with a massless scalar field gets all fields massive with equal masses and the couplings proper to a supersymmetric model.
\end{theorem}

\begin{proof}
We can write down the equations of motion for the theory obtaining
\begin{eqnarray}
   \partial^2A_\mu-\partial_\mu(\partial\cdot A)&=&
   e\bar\psi\gamma_\mu\psi+ie\left(\phi^\dagger\partial_\mu\phi-\partial_\mu\phi^\dagger\phi\right)-e^2|\phi|^2A_\mu \\ \nonumber
   \left(i\slashed{\partial}-e\slashed{A}-g\phi\right)\psi&=&0 \\ \nonumber
   \partial^2\phi+\lambda|\phi|^2\phi&=&2ie\partial^\mu\phi A_\mu+e^2A^2\phi+ie\phi\partial\cdot A
\end{eqnarray}
The choice of the gauge is not so relevant as we will see below, so we choose Lorenz gauge. Here we assume, and will prove, that $g,e=o(\lambda)$ and so we can assume an ordering, in the limit $\lambda\rightarrow\infty$, so to have to solve at the leading order the following set of equations
\begin{eqnarray}
\label{eq:set}
   \partial^2\phi_0+\lambda|\phi_0|^2\phi_0&=&2ie\partial^\mu\phi_0 A_\mu^{(0)}+e^2\phi_0[A^{(0)}]^2 \\ \nonumber
   (i\slashed{\partial}-e\slashed{A}^{(0)}-g\phi_0)\psi_0&=&0 \\ \nonumber
   \partial^2A_\mu^{(0)}+e^2|\phi_0(\xi)|^2A_\mu^{(0)}&=&
   ie\left(\phi^\dagger_0\partial_\mu\phi_0-\partial_\mu\phi^\dagger_0\phi_0\right).
\end{eqnarray}
We will compute the next-to-leading order equation for the gauge field below as it will be used for the proof of the theorem.

In order to check the correctness of the eqs.(\ref{eq:set}), we take the following series
\begin{eqnarray}
   \phi&=&\phi_0+\frac{1}{\sqrt{\lambda}}\phi_1+O\left(\frac{1}{\lambda}\right) \\ \nonumber
   \psi&=&\psi_0+\frac{1}{\sqrt{\lambda}}\psi_1+O\left(\frac{1}{\lambda}\right) \\ \nonumber
   A_\mu&=&A_\mu^{(0)}+\frac{1}{\sqrt{\lambda}}A_\mu^{(1)}+O\left(\frac{1}{\lambda}\right). 
\end{eqnarray}
Then, we operate the following rescaling on the space-time variables, for the scalar field equation
\begin{equation}
   x_\mu\rightarrow \sqrt{\lambda}x_\mu
\end{equation}
so, one has
\begin{equation}
   \partial^2\phi+|\phi|^2\phi=2i\frac{e}{\sqrt{\lambda}}\partial^\mu\phi A_\mu+\frac{e^2}{\lambda}A^2\phi
\end{equation}
and assuming that $e\sim\sqrt{\lambda}$ to be checked {\sl a posteriori}, gives the leading order equation
\begin{equation}
   \partial^2\phi_0+|\phi_0|^2\phi_0=2i\frac{e}{\sqrt{\lambda}}\partial^\mu\phi_0 A_\mu^{(0)}+\frac{e^2}{\lambda}\phi_0[A^{(0)}]^2.
\end{equation}
We operate in the same way on the equation of the gauge field that, after rescaling, becomes
\begin{equation}
\label{eq:res}
  \partial^2A_\mu+\frac{e^2}{\lambda}|\phi|^2A_\mu=
   \frac{e}{\lambda}\bar\psi\gamma_\mu\psi+i\frac{e}{\sqrt{\lambda}}\left(\phi^\dagger\partial_\mu\phi-\partial_\mu\phi^\dagger\phi\right)
\end{equation}
that will give at the leading order, always assuming $e\sim\sqrt{\lambda}$,
\begin{equation}
  \partial^2A_\mu^{(0)}+\frac{e^2}{\lambda}|\phi_0|^2A_\mu^{(0)}=
   i\frac{e}{\sqrt{\lambda}}\left(\phi^\dagger_0\partial_\mu\phi_0-\partial_\mu\phi^\dagger_0\phi_0\right).
\end{equation}
Finally, for the Dirac field one has, after rescaling,
\begin{equation}
  \left(i\slashed{\partial}-\frac{e}{\sqrt{\lambda}}\slashed{A}-\frac{g}{\sqrt{\lambda}}\phi\right)\psi=0 
\end{equation}
giving at the leading order, assuming also $g\sim\sqrt{\lambda}$ to be checked {\sl a posteriori},
\begin{equation}
  \left(i\slashed{\partial}-\frac{e}{\sqrt{\lambda}}\slashed{A}^{(0)}-\frac{g}{\sqrt{\lambda}}\phi_0\right)\psi_0=0. 
\end{equation}
Next-to-leading order will be, starting from the rescaled equation (\ref{eq:res}),
\begin{eqnarray}
   \partial^2A_\mu^{(1)}+e^2|\phi_0|^2A_\mu^{(1)}&=&
   e\bar\psi_0\gamma_\mu\psi_0+ie\left(\phi_1^\dagger\partial_\mu\phi_0-\partial_\mu\phi_1^\dagger\phi_0\right)\\ \nonumber
   &+&
   ie\left(\phi_0^\dagger\partial_\mu\phi_1-\partial_\mu\phi_0^\dagger\phi_1\right)
   -e^2\phi_0^\dagger\phi_1A_\mu^{(0)}-e^2\phi_1^\dagger\phi_0A_\mu^{(0)}.
\end{eqnarray}
So, we see that our ordering exists and we can look for a solution.

The set of equations (\ref{eq:set}) can be easily solved if we take $A_\mu^{(0)}=0$. So, one has immediately
\begin{equation}
   \phi_0=\mu\left(\frac{2}{\lambda}\right)^\frac{1}{4}e^{i\alpha}{\rm sn}(n\cdot x+\theta,i)
\end{equation}
being $\rm sn$ a Jacobi elliptical function, $\alpha$ and $\theta$ two constant phases and $\mu$ an integration constant. This holds provided
\begin{equation}
\label{eq:disp}
   n^2=\mu^2\sqrt{\frac{\lambda}{2}}
\end{equation}
that is, the scalar field, at the leading order, describes a free massive particle. It is easy to verify that $\phi_0^\dagger\partial_\mu\phi_0-\partial_\mu\phi^\dagger_0\phi_0=0$ and this solution is perfectly consistent.

Similarly, for the Fermion field one gets
\begin{equation}
   \psi_0=[{\rm dn}(\xi,i)-{\rm cn}(\xi,i)]^{e^{i\alpha}g\sqrt{\frac{2}{\lambda}}} u
\end{equation}
provided $\slashed{n}u=mu$. Here $\xi=n\cdot x+\theta$ and $\rm dn$ and $\rm cn$ are Jacobi elliptic functions. So, by the dispersion relation (\ref{eq:disp}), we must conclude that the Fermion and the scalar fields must have the same mass at the leading order. Finally, if we want that the wave function of the Fermion is single valued we must also take $g=\sqrt{\frac{\lambda}{2}}$ and $\alpha=2\pi m$ with $m$ an integer, consistently with our ordering arguments.

Finally, we consider the next-to-leading order equation for the gauge field. One can express this through the variable $\xi$ defined above. Using the dispersion relation (\ref{eq:disp}) and the above solutions, this takes the form
\begin{eqnarray}
    \frac{d^2A_\mu^{(1)}}{d\xi^2}+\frac{2e^2}{\lambda}{\rm sn}^2(\xi,i)A_\mu^{(1)}&=&
    \sqrt{\frac{2e^2}{\mu^4\lambda}}\bar\psi_0\gamma_\mu\psi_1+\sqrt{\frac{2e^2}{\mu^4\lambda}}\bar\psi_1\gamma_\mu\psi_0 \\ \nonumber
   &+&i\sqrt{\frac{2e^2}{\mu^4\lambda}}\left(\phi_1^\dagger\partial_\mu\phi_0-\partial_\mu\phi_1^\dagger\phi_0\right)
   +i\sqrt{\frac{2e^2}{\mu^4\lambda}}\left(\phi_0^\dagger\partial_\mu\phi_1-\partial_\mu\phi_0^\dagger\phi_1\right)
\end{eqnarray}
The general equation
\begin{equation}
   y''(x)+a\cdot{\rm sn}^2(b\cdot x,i)y(x)=\delta(x)
\end{equation}
and the corresponding homogeneous equation admit stable solutions only for $b^2=a/6$ so, in order to write down a physical solution we must have $e^2=3\lambda$. The Green function can be written down
\begin{equation}
    G(\xi)=-\frac{1}{4}\theta(\xi){\rm cn}(\xi+\varphi,i){\rm dn}(\xi+\varphi,i)
\end{equation}
being $\varphi$ any value such that ${\rm cn}(\varphi,i)=0$ and the homogeneous equation has the simple solution
\begin{equation}
    A_\mu^{hom}=\epsilon_\mu {\rm cn}(\xi,i){\rm dn}(\xi,i)
\end{equation}
being $\epsilon_\mu$ a constant vector. These solutions describe a free massive particle with the same mass of the scalar and Fermion fields and a proper coupling consistent with supersymmetry.

We see that, in order to have a mechanism to generate mass, the model must have all the properties of supersymmetry. This completes the proof.

\end{proof}

An interesting consequence of this theorem is that, provided the mechanism generating masses through a scalar field is the same for quantum chromodynamics and standard model, necessarily the theory must imply supersymmetry and the its breaking. Conversely, it is not needed to introduce mass terms into a supersymmetric model as masses can be dynamically generated through scalar superfields.



\section{Standard Model\label{sec3}}

A few considerations are in order for the standard model. The very nature of this mechanism implies that we need a number of scalar fields in order to have all the particles of the theory massive. So, this means that, if we want to adopt this mechanism also in the standard model, all Fermion fields will have identical masses and the same happens for bosons. An immediate implication of this simple conclusion based on the theorem proved above is that the only way to get all this matter to have a physical meaning is by considering supersymmetric generalizations of the standard model. The only difference we expect is the presence of just the quartic term for the scalar fields and to assume such a self-interaction becoming increasingly large. But, as the gauge coupling will be fixed through the self-interaction of these scalar fields, it is important to compute the corresponding beta function for the coupling. We showed that the coupling of the scalar field decreases to zero lowering momenta\cite{fra4} making all this scenario consistent. This should be expected on the basis of the triviality of the theory. But, modifications of this beta function with the scalar field interacting with other fields has to be computed.

An immediate consequences of this results is that the observation of the Higgs particle may be not enough to understand the proper mechanism of mass generation. Also, to understand the form of the corresponding potential of the scalar fields is essential to tell what is the right mechanism at work.

Finally, we note that this mechanism, {\sl per se}, is able to give to all the particles of the theory an identical mass while the coupling are also properly fixed. So, mass differences can only be understood through the mechanism that eventually breaks supersymmetry. This is the way this kind of models are currently built.

\section{Conclusions\label{sec4}}

We have seen how a Higgs particle can produce masses to all the fields interacting with it, being massless but provided its self-interaction is very large. This same mechanism is at work in a Yang-Mills field to produce a massive field already at classical level. So, if one should expect that the same mechanism is at work both for strong interactions and for the masses of the particle in the standard model, it is essential to get an understanding of the form of the self-interaction of the Higgs field. But if this should be the same mechanism for both then supersymmetry must enter into play. This is our fundamental conclusion. 

Presently, we limited our analysis just to a classical level. It is possible to extend this to a quantum field theory. We aim to do this in a very near future. But we point out how, already at the classical level, the result is quite unexpected.


\end{document}